\documentclass[12pt,DIV10,final]{scrartcl}
\usepackage{amsmath}
\usepackage{amssymb}
\usepackage{amsthm}
\usepackage{enumerate}
\usepackage{enumitem}
\usepackage[english]{babel}
\usepackage{ifthen}
\usepackage[notref,notcite,color]{showkeys}
\usepackage{multirow}
\usepackage{xspace}
\usepackage{url}

\usepackage[final]{hyperref}

\usepackage{tikz}

\definecolor{darkgreen}{rgb}{0.0,0.7,0.0}
\newenvironment{tw}{\noindent\color{darkgreen} TW:}{}

\newenvironment{mk}{\noindent\color{blue} MK:} {}

\newenvironment{vd}{\noindent\color{red} VD:} {}

\newcommand{\refthm}[1]{Theorem~\ref{#1}\xspace}
\newcommand{\refcor}[1]{Corollary~\ref{#1}\xspace}

\newcommand{\reflem}[1]{Lemma~\ref{#1}\xspace}

\newcommand{\refclm}[1]{Claim~\ref{#1}\xspace}

\newcommand{\IFF}{if and only if\xspace}
\newcommand{\homo}{homomorphism\xspace}

\newcommand{\lr}{length-redu\-cing\xspace}

\newcommand{\set}[2]{\left\{#1\mathrel{\left|\vphantom{#1}\vphantom{#2}\right.}#2\right\}}
\newcommand{\oneset}[1]{\left\{\mathinner{#1}\right\}}
\newcommand{\smallset}[1]{\left\{#1\right\}}

\newcommand{\abs}[1]{\left|\mathinner{#1}\right|}
\newcommand{\Abs}[1]{\left\lVert\mathinner{#1}\right\rVert}
\newcommand{\floor}[1]{\left\lfloor\mathinner{#1} \right\rfloor}

\newcommand{\N}{\mathbb{N}}
\newcommand{\Z}{\mathbb{Z}}

\newcommand{\sse}{\subseteq}

\newcommand{\IRR}{\mathrm{IRR}}

\renewcommand{\phi}{\varphi}

\newcommand\RAS[2]{\overset{#1}{\underset{#2}{\Longrightarrow}}}

\newcommand\LAS[2]{\overset{#1}{\underset{#2}{\Longleftarrow}}}
\newcommand\DAS[2]{\overset{#1}{\underset{#2}{\Longleftrightarrow}}}

\newcommand\RA[1]{\underset{#1}{\Longrightarrow}}
\newcommand\LA[1]{\underset{#1}{\Longleftarrow}}

\newtheorem{theorem}{Theorem}
\newtheorem{definition}[theorem]{Definition}

\newtheorem{lemma}[theorem]{Lemma}
\newtheorem{corollary}[theorem]{Corollary}

\newtheorem{claim}{Claim}

\newtheorem{expl}[theorem]{Example}

\setlist{itemsep=3pt,parsep=3pt,topsep=4pt}

\begin{document}

\title{Regular Languages are {Church-Rosser} Congruential}

\author{%
  Volker Diekert\,$^{*}$ \\
  Manfred Kuf\-leitner\,$^{*}$
  \\
  \,\ Klaus Reinhardt\,$^{\dagger}$
  \\
  \,Tobias Walter\,$^{*}$ \\[4mm]
  {\footnotesize \!\!$^*$\;Institut 
    f\"ur Formale Methoden der Informatik} \\[-3.5mm]
  {\footnotesize University of Stuttgart, Germany} \\[0mm]
  {\footnotesize \!\!$^{\dagger}$\;Wilhelm-Schickard-Institut f{\"u}r 
    Informatik} \\[-3.5mm]
  {\footnotesize University of T{\"u}bingen, Germany}
}


\date{}

\maketitle

\begin{abstract}
  \noindent
  \textbf{Abstract.}\,
  This paper proves a long standing conjecture in formal language
  theory.  It shows that all regular languages are Church-Rosser
  congruential.  The class of Church-Rosser congruential languages was
  introduced by \mbox{McNaughton}, Narendran, and Otto in 1988. A
  language $L$ is Church-Rosser congruential, if there exists a finite
  confluent, and length-reducing semi-Thue system $S$ such that $L$ is
  a finite union of congruence classes modulo $S$.  It was known that
  there are deterministic linear context-free languages which are not
  Church-Rosser congruential, but on the other hand it was
  strongly believed that all regular language are of this form.
  Actually, this paper proves a more general result.\,\footnote{ The
    research on this paper was initiated during the program
    \emph{Automata Theory and Applications} at the Institute for
    Mathematical Sciences, National University of Singapore in
    September~2011. The second author was supported by the German
    Research Foundation (DFG) under grant \mbox{DI 435/5-1}.}

  \medskip

\end{abstract}

\section{Introduction}\label{sec:intro}

It has been a long standing conjecture in formal language theory that
all regular languages are Church-Rosser congruential. The class of
Church-Rosser congruential languages was introduced by
\mbox{McNaughton}, Narendran, and Otto
in~1988~\cite{McNaughtonNO88}. A language $L$ is Church-Rosser
congruential, if there exists a finite confluent, and length-reducing
semi-Thue system $S$ such that $L$ is a finite union of congruence
classes modulo $S$. One of the main motivations to consider this class
of languages is that the membership problem for $L$ can be solved in
linear time; this is done by computing normal forms using the system
$S$, followed by a table look-up. For this it is not necessary that
the quotient monoid $A^*/S$ is finite, it is enough that $L$ is a
finite union of congruence classes modulo $S$. It is not hard to see
that $\set{a^nb^n}{n\in \N}$ is Church-Rosser congruential, but
$\set{a^mb^n}{m,n\in \N \text{ and } m \geq n }$ is not.  This led the
authors of~\cite{McNaughtonNO88} to the more technical notion of
Church-Rosser languages; this class of languages captures all
deterministic context-free languages. For more results about
Church-Rosser languages see
e.g.~\cite{BuntrockO98,Narendran84phd,Woinowski01Diss,Woinowski03IC}.

From the very beginning it was strongly believed that all regular
languages are Church-Rosser congruential in the pure sense.  However,
after some significant initial
progress~\cite{Narendran84phd,NiemannPhD02,NiemannO05,NiemannW02,reinhardtT03}
there was some stagnation.
 
Before 2011 the most advanced result was the one announced in 2003 by
Reinhardt and Th\'erien \cite{reinhardtT03}.  According to this
manuscript the conjecture is true for all regular languages where the
syntactic monoid is a group. However, the manuscript has never been
published as a refereed paper and there are some flaws in its
presentation.  The main problem with~\cite{reinhardtT03} has however
been quite different for us. The statement is too weak to be useful in
the induction for the general case. So, instead of being able to
use~\cite{reinhardtT03} as a black box, we shall prove a more general
result in the setting of weight-reducing systems. This part about
group languages is a cornerstone in our approach.

The other ingredient to our paper has been established only very
recently.  Knowing that the result is true if the the syntactic monoid
is a group, we started looking at aperiodic monoids. Aperiodic monoids
correspond to star-free languages and the first two authors together
with Weil proved that all star-free languages are Church-Rosser
congruential~\cite{DiekertKW12tcs}.  Our proof became possible by
\emph{loading the induction hypothesis}. This means we proved a much
stronger statement.  We showed that for every star-free language $L
\sse A^*$ there exists a finite confluent semi-Thue system $S\sse A^*
\times A^*$ such that the quotient monoid $A^*/ S$ is finite (and
aperiodic), $L$ is a union of congruence classes modulo $S$, and
moreover all right-hand sides of rules appear as scattered subwords in
the corresponding left-hand side. We called the last property
\emph{subword-reducing}, and it is obvious that every subword-reducing
system is length-reducing.

We have little hope that such a strong result could be true in
general. Indeed here we step back from subword-reducing to
weight-reducing systems.

We prove in \refthm{thm:main} the following result: Let $L \sse A^*$
be a regular language and $\Abs{a} \in \N \setminus \smallset{0}$ be a
positive weight for every letter $a \in A$ (e.g.,{} $ \Abs{a} =
\abs{a} =1$).  Then we can construct for the given weight a finite,
confluent and weight-reducing semi-Thue system $S\sse A^* \times A^*$
such that the quotient monoid $A^*/ S$ is finite and recognizes
$L$. In particular, $L$ is a finite union of congruence classes modulo
$S$.

Note that this gives us another characterization for the class of
regular languages.  By \refcor{cor:main} we see that a language $L
\subseteq A^*$ is regular if and only if $L$ is recognized by a finite
Church-Rosser system $S$ with finite index.  As a consequence, a long
standing conjecture about regular languages has been solved
positively.

\section{Preliminaries}\label{sec:prelim}

\paragraph{Words and languages}

Throughout this paper, $A$ is a finite alphabet. An element of $A$ is
called a \emph{letter}. The set $A^*$ is the free monoid generated by
$A$. It consists of all finite sequences of letter from $A$. The
elements of $A^*$ are called \emph{words}. The empty word is denoted
by $1$. The \emph{length} of a word $u$ is denoted by $\abs{u}$. We
have $\abs{u} = n$ for $u= a_1 \cdots a_n$ where $a_i \in A$. The
empty word has length $0$, and it is the only word with this property.
The set of word of length at most $n$ is denoted by $A^{\leq n}$, and
the set of all nonempty words is $A^+$.  We generalize the length of a
word by introducing weights. A \emph{weighted alphabet}
$(A,\Abs{\cdot})$ consists of an alphabet $A$ equipped with a weight
function $\Abs{\cdot} : A \to \N \setminus \smallset{0}$. The
\emph{weight} of a letter $a \in A$ is $\Abs{a}$ and the \emph{weight}
$\Abs{u}$ of a word $u= a_1 \cdots a_n$ with $a_i \in A$ is $\Abs{a_1}
+ \cdots + \Abs{a_n}$. The weight of the empty word is $0$. The length
is the special weight with $\Abs{a} = 1$ for all $a \in A$.  A word
$u$ is a \emph{factor} of a word $v$ if there exist $p,q \in A^*$ such
that $puq = v$, and $u$ is a \emph{proper factor} of $v$ if $pq \neq
1$. The word $u$ is a \emph{prefix} of $v$ if $uq = v$ for some $q \in
A^*$, and it is a \emph{suffix} of $v$ if $pu = v$ for some $p \in
A^*$.  We say that $u$ is a factor (resp.\ prefix, resp.\ suffix) of
$v^+$ if there exists $n \in \N$ such that $u$ is a factor (resp.\
prefix, resp.\ suffix) of $v^n$. Two words $u,v \in A^*$ are
\emph{conjugate} if there exist $p,q \in A^*$ such that $u = pq$ and
$v = qp$. An integer $m > 0$ is a \emph{period} of a word $u = a_1
\cdots a_n$ with $a_i \in A$ if $a_i = a_{i+m}$ for all $1 \leq i \leq
n-m$. A word $u \in A^+$ is \emph{primitive} if there exists no $v \in
A^+$ such that $u = v^n$ for some integer $n > 1$. It is a standard
fact that a word $u$ is not primitive if and only $u^2 = puq$ for some
$p,q \in A^+$. This follows immediately from the result from
combinatorics on words that $xy = yx$ \IFF $x$ and $y$ are powers of a
common root; see e.g.~\cite[Section~1.3]{lot83}.

A monoid $M$ \emph{recognizes} a language $L \subseteq A^*$ if there
exists a homomorphism $\varphi : A^* \to M$ such that $L =
\varphi^{-1} \varphi(L)$.  A language $L \subseteq A^*$ is
\emph{regular} if it is recognized by a finite monoid.  There are
various other 
and well-known characterizations of regular languages; e.g., regular
expressions, finite automata or monadic second order logic.  Regular
languages $L$ can be classified in terms of structural properties of
the monoids recognizing $L$. In particular, we consider group
languages; these are languages recognized by finite groups.

\paragraph{Semi-Thue systems}

A {\em semi-Thue system} over $A$ is a subset $S\sse A^*\times
A^*$. In this paper, all semi-Thue systems are finite. The elements of
$S$ are called {\em rules}.  We frequently write $\ell \to r$ for
rules $(\ell,r)$.  A system $S$ is called {\em \lr} if we have $\abs
\ell > \abs r $ for all rules $\ell \to r$ in $S$. It is called
\emph{weight-reducing} 
with respect to some weighted alphabet $(A,\Abs{\cdot})$, if
$\Abs{\ell} > \Abs{r}$ for all rules $\ell \to r$ in $S$.  Every
system $S$ defines the rewriting relation $\RA{S} \sse A^* \times A^*$
by setting $u \RA{S} v$ if there exist $p,q,\ell,r \in A^*$ such that
$u = p\ell q$, $v= prq$, and $\ell \to r$ is in $S$.
 
By $\RAS*{S}$ we mean the reflexive and transitive closure of
$\RA{S}$.  By $\DAS*{S}$ we mean the symmetric, reflexive, and
transitive closure of $\RA{S}$. We also write $u \LAS*{S}v$ whenever
$v\RAS*{S}u$.  The system $S$ is {\em confluent} if for all
$u\DAS*{S}v$ there is some $w$ such that $u\RAS*{S}w\LAS*{S}v$.  It is
\emph{locally confluent} if for all $v \LA{S} u \RA{S} v'$ there
exists $w$ such that $v \RAS*{S} w \LAS*{S} v'$. If $S$ is locally
confluent and weight-reducing for some weight, then $S$ is confluent;
see e.g.~\cite{bo93springer,jan88eatcs}.  Note that $u \RAS{}S v$ implies that
$\Abs{u} > \Abs{v}$ for weight-reducing systems.  The relation
${\DAS*{S}} \sse A^* \times A^*$ is a congruence, hence the congruence
classes $[u]_S = \{v \in A^*\mid u \DAS*S v\}$ form a monoid which is
denoted by $A^*/S $. The size of $A^*/S$ is called the \emph{index of
  $S$}.  A finite semi-Thue system $S$ can be viewed as a finite set
of defining relations. Hence, $A^*/S $ becomes a finitely presented
monoid.  By $\IRR_S(A^*)$ we denote the set of irreducible words in
$A^*$, i.e., the set of words where no left-hand side occurs as a
factor.

Whenever the weighted alphabet $(A,\Abs{\cdot})$ is fixed, a finite
semi-Thue system $S \subseteq A^* \times A^*$ is called a
\emph{weighted Church-Rosser system} if it is finite, weight-reducing
for $(A,\Abs{\cdot})$, and confluent.  Hence, a finite semi-Thue
system $S$ is a weighted Church-Rosser system \IFF (1) we have $\Abs
\ell > \Abs r$ for all rules $\ell \to r$ in $S$ and (2) every
congruence class has exactly one irreducible element.  In particular,
for weighted Church-Rosser systems $S$, there is a one-to-one
correspondence between $A^* / S$ and $\IRR_S(A^*)$.
A \emph{Church-Rosser system} is a finite, length-reducing, and
confluent semi-Thue system. In particular, every Church-Rosser system
is a weighted Church-Rosser system.  A language $L \sse A^*$ is called
a {\em Church-Rosser congruential language} if there is a finite
Church-Rosser system $S$ such that $L$ can be written as a finite
union of congruence classes $[u]_S$.

\begin{definition}\label{def:cr}
  Let $\phi: A^* \to M$ be a \homo and let $S$ be a semi-Thue
  system. We say that \emph{$\phi$ factorizes through $S$} if for all
  $u,v \in A^*$ we have:
  \begin{equation*}
    u \DAS*{S} v \quad \text{implies} \quad \varphi(u) = \varphi(v).
  \end{equation*}
\end{definition}

Note that if $S$ is a semi-Thue system and $\phi : A^* \to M$
factorizes through $S$, then the following diagram commutes:
\begin{center}
  \begin{tikzpicture}[scale=0.75]
    \draw (0,0) node (A) {$A^*$};
    \draw (3,0) node (M) {$M$};
    \draw (3,2) node (S) {$A^* / S$};
    \draw[->] (A) -- node[pos=0.65,above] {$\varphi$} (M);
    \draw[->] (A) -- node[above left,outer sep=-2pt] {$\pi$} (S);
    \draw[->] (S) -- node[right] {$\psi$} (M);
  \end{tikzpicture}
\end{center}
Here, $\pi(u) = [u]_S$ is the canonical homomorphism and $\psi([u]_S)
= \varphi(u)$.

\section{Finite Groups}\label{sec:groups}

Our main result is that every homomorphism $\varphi : A^* \to M$ to
finite monoid $M$ factorizes through a  Church-Rosser system
$S$. Our proof of this theorem distinguishes whether or not $M$
is a group. Thus, we first prove this result for groups. Before we
turn to the general case, we show that for some particular groups,
proving the claim is easy. The techniques developed here will also be
used when proving the result for arbitrary finite groups.

\subsection{Groups without proper cyclic quotient groups}\label{sec:simple}

The aim of this section is to show that finding a Church-Rosser
system is very easy for many cases. This list includes
systems of all finite (non-cyclic) simple groups, but it goes far
beyond this. Let $\phi: A^* \to G$ be a homomorphism to a finite
group, where $(A,\Abs{\cdot})$ is a weighted alphabet. This defines a
regular language $L_G = \set{w \in A^*}{\phi(w) = 1}$. Let us assume
that the greatest common divisor $\gcd{\set{\Abs{w}}{w \in L_G}}$ is
equal to one; e.g.{} $\oneset{6,10,15} \sse \set{\Abs{w}}{w \in
  L_G}$. Then there are two words $u,v \in L_G$ such that $\Abs{u}
-\Abs{v} = 1$. Now we can use these words to find a constant $d$ such
that all $g\in G$ have a representing word $v_g$ with the exact weight
$\Abs{v_g} = d$. To see this, start with some arbitrary set of
representing words $v_g$. We multiply words $v_g$ with smaller weight
with $u$ and words $v_g$ higher weights with $v$ until all weights are
equal.

The final step is to define the following weight-reducing system
$$S_G= \set{w \to v_{\phi(w)}} {w \in A^{*}\;  \text{ and } d < \Abs w \leq d + \max\set{\Abs a}{a \in A}}.$$
Confluence of $S_G$ is trivial; and every language recognized by
$\phi$ is also recognized by the canonical \homo $A^* \to A^* / S_G$.

Now assume that we are not so lucky, i.e., $\gcd{\set{\Abs{w}}{w \in
    L_G}} >1$.  This means there is a prime number $p$ such that $p$
divides $\Abs{w}$ for all $w \in L_G$. Then, the \homo of $A^*$ to
$\Z/p\Z$ defined by $a \mapsto \Abs{a} \bmod p$ factorizes through
$\phi$ and $\Z/p\Z$ becomes a quotient group of $G$.  This can never
happen if $G$ is simple and non-cyclic, because a simple group does
not have any proper quotient group. But there are many other cases
where a natural homomorphism $A^* \to G$ for some weighted alphabet
$(A,\Abs{\cdot})$ satisfies the property $\gcd{\set{\Abs{w}}{w \in
    L_G}= 1}$ although $G$ has a non-trivial cyclic quotient group.
Just consider the length function and a presentation by standard
generators for dihedral groups $D_{2n}$ or the permutation groups
$\mathcal{S}_n$ where $n$ is odd.

For example, let $G = D_6 = \mathcal{S}_3$ be the permutation group of
a triangle.  Then~$G$ is generated by elements $\tau$ and $\rho$ with
defining relations
\begin{equation*}
  \tau^2 = \rho^3 = 1 \text { and } \tau \rho \tau = \rho^2.
\end{equation*}
The following six words of length $3$ represent all six group
elements: 
\begin{equation*}
  1=\rho^3,\; \rho=\rho\tau^2,\; \rho^2 = \tau \rho \tau,\; \tau=\tau^3,\;
  \tau\rho = \rho^2\tau,\;  \tau\rho^2.
%
\end{equation*}
The corresponding monoid $\smallset{\rho,\tau}^* / S_G$ has $15$ elements.

It is much harder to find a Church-Rosser system for the \homo
$\phi: \smallset{a,b,c}^* \to \Z / 3 \Z$ where $\phi(a) = \phi(b) =
\phi(c) = 1 \bmod 3$.  In some sense this phenomenon suggests that 
finite cyclic groups or more general commutative groups are the
obstacle to find a simple construction for Church-Rosser
systems.

\subsection{The general case for group languages}

In this section, we consider arbitrary groups. We start with some
simple properties of Church-Rosser systems. Then, in
\refthm{thm:group}, we state and prove that group languages are
Church-Rosser congruential.


\begin{lemma}\label{lem:irr}
  Let $(A,\Abs{\cdot})$ be a weighted alphabet, let $d \in \N$, and
  let $S \subseteq A^* \times A^*$ be a weighted Church-Rosser system
  such that $\IRR_S(A^*)$ is finite. Then
  \begin{equation*}
    S_d = \set{u \ell v \to u r v}{u,v \in A^d \text{ and } \ell \to r \in S}
  \end{equation*}
  is a weighted Church-Rosser system satisfying:
  \begin{enumerate}
  \item\label{irri} $\IRR_{S_d}(A^*)$ is finite. 
  \item\label{irrii} All words of length at most $2d$ are irreducible
    with respect to $S_d$.
  \item\label{irriii} The mapping $[u]_{S_d} \mapsto [u]_S$ for $u \in
    A^*$ is well-defined and yields a surjective homomorphism from
    $A^* / S_d$ onto $A^* / S$.
  \end{enumerate}
\end{lemma}

\begin{proof}
  First, one shows that local confluence of $S$ transfers to local
  confluence of $S_d$. For ``\ref{irri}'' and ``\ref{irrii}'' note that
  $\IRR_{S_d}(A^*) = A^{\leq 2d} \cup A^d \cdot \IRR_S(A^*) \cdot A^d$.
  The remaining proof is straightforward and therefore left to the
  reader.
  %
\end{proof}

\begin{lemma}\label{lem:econfl}
  Let $(A,\Abs{\cdot})$ be a weighted alphabet and let $\Delta
  \subseteq A^+$ such that all words in $\Delta$ have length at most
  $t$. Then, for every $n \geq 1$, the set of rules
  \begin{equation*}
    T = \set{\delta^{t+n} \to \delta^t}{%
      \delta \in \Delta,\; \delta \text{ is primitive}}
  \end{equation*}
  yields a weighted Church-Rosser system.
\end{lemma}

\begin{proof}
  Every rule in $T$ is weight-reducing. Thus it suffices to show that
  $T$ is locally confluent. Let $\delta, \tilde{\delta} \in \Delta$ be
  primitive with $\abs{\delta} \geq |\tilde{\delta}|$ and suppose $x
  \delta^{t+n} = \tilde{\delta}^{t+n} y$. If $\delta^{t+n}$ is a
  suffix of $\tilde{\delta}^{t} y$, then $\tilde{\delta}^{t+n}$ is a
  prefix of $x \delta^t$; and the two $T$-rules $\delta^{t+n} \to
  \delta^t$ and $\tilde{\delta}^{t+n} \to \tilde{\delta}^t$ can be
  applied independently of one another. Thus we can assume
  $\abs{\delta^{t+n}} > |\tilde{\delta}^{t} y|$. In particular,
  $\tilde{\delta}^t$ is a factor of $\delta^+$. Note that
  $|\tilde{\delta}^t| \geq \abs{\delta}$. Thus $|\tilde{\delta}|$ is
  a period of $\delta$.

  Let us first consider the case $\abs{\delta} > |\tilde{\delta}|$.
  Since $\delta$ is primitive, $|\tilde{\delta}|$ cannot be a divisor
  of $\abs{\delta}$. In particular, we have $|\tilde{\delta}| \geq 2$.
  Suppose $|\tilde{\delta}| = 2$. Then $\delta = (ab)^ma$ for $a,b \in
  A$ and some $m \geq 1$. We conclude that the suffix $a\delta$ or the
  prefix $\delta a$ of $\delta^2$ is a factor of
  $\tilde{\delta}^+$. Since both words $a\delta$ and $\delta a$ have a
  factor $aa$ and $|\tilde{\delta}| = 2$, this contradicts
  $\tilde{\delta}$ being primitive. Therefore, we can assume
  $|\tilde{\delta}| \geq 3$ and hence, $|\tilde{\delta}^t| \geq
  \abs{\delta^3}$. It follows that $\delta^2$ is a factor of
  $\tilde{\delta}^+$ and $|\tilde{\delta}|$ is a period of $\delta^2$.
  By shifting the prefix $\delta$ of $\delta^2$ by this period, we can
  write $\delta^2 = p \delta q$ with $p,q \in A^+$ and $\abs{p} =
  |\tilde{\delta}|$.  We conclude that $\delta$ is not primitive,
  which is a contradiction.

  Let now $\abs{\delta} = |\tilde{\delta}|$. In this case, the words
  $\delta$ and $\tilde{\delta}$ are conjugate.  Therefore, applying
  one of the rules $\delta^{t+n} \to \delta^t$ and
  $\tilde{\delta}^{t+n} \to \tilde{\delta}^t$ yields the same word. 
  %
\end{proof}
\begin{lemma}\label{lem:efin}
  Let $\Delta \subseteq A^+$ be a set of words such that all words in
  $\Delta$ have length at most $n$. If $u \in A^{> 2n}$ is not a
  factor of some $\delta^+$ for $\delta \in \Delta$, then there is a
  proper factor $v$ of $u$ which is also not a factor of some
  $\delta^+$ for $\delta \in \Delta$.
\end{lemma}

\begin{proof}
  Assume that such a factor $v$ of $u$ does not exist. Let $u = a w b$
  for $a,b \in A$. Then $aw$ is a factor of $\delta^+$ and $wb$ is a
  factor of $\delta'^+$ for some $\delta,\delta' \in \Delta$. Let $p =
  \abs{\delta}$ and $q = \abs{\delta'}$. Now, $p$ is a period of $aw$
  and $q$ is a period of $wb$. Thus $p$ and $q$ are both periods of
  $w$.  Since $\abs{w} \geq 2n-1 \geq p + q - \gcd(p,q)$, we see that
  $\gcd(p,q)$ is also a period of $w$ by the Periodicity Lemma of Fine
  and Wilf~\cite[Section~1.3]{lot83}. The $(p+1)$-th letter in $aw$ is $a$.
  Going in steps $\gcd(p,q)$ to the left or to the right in $w$, we see
  that the $(q+1)$-th letter in $aw$ is $a$. Thus $awb$ is
  a factor of $\delta'^+$, which is a contraction.
\end{proof}

We are now ready to prove the main result of this section: Group
languages are Church-Rosser congruential. An outline of the proof is
as follows. By induction on the size of the alphabet, we show that
every homomorphism $\varphi : A^* \to G$ factorizes through a weighted
Church-Rosser system $S$ with finite index. Remove some letter $c$
from the alphabet $A$. This leads to a system $R$ for the remaining
letters $B$. \reflem{lem:irr} allows to assume that certain words are
irreducible. Then we consider $K = \IRR_R(B^*) c$ which is a prefix
code in $A^*$. We consider $K$ as a new alphabet. Essentially, it is
this situation where weighted alphabets come into play because we can
choose the weight of $K$ such that it is compatible with the weight
over the alphabet $A$. Over $K$, we introduce two sets of rules
$T_\Delta$ and $T_\Omega$. The $T_\Delta$-rules reduce long
repetitions of short words $\Delta$, and the $T_\Omega$-rules have the
form $\omega \, u \, \omega \to \omega \, v_g \, \omega$. Here,
$\Omega$ is some finite set of markers and $\omega \in \Omega$ is such
a marker.  The word $v_g$ is a normal form for the group
element~$g$. The $T_\Omega$-rules reduce long words without long
repetitions of short words.  Then we show that $T_\Delta$ and
$T_\Omega$ are confluent and that their union has finite index
over~$K^*$. Here, the confluence of the $T_\Delta$-rules is
\reflem{lem:econfl}.  The confluence of the $T_\Omega$-rules relies on
several combinatorial properties of the normal forms $v_g$ and the
markers $\Omega$. Using \reflem{lem:efin}, we see that all
sufficiently long words are reducible.  Since by construction all
rules in $T = T_\Delta \cup T_\Omega$ are weight-reducing, the system
$T$ is a weighted Church-Rosser system over~$K^*$ with finite
index such $\varphi : K^* \to G$ factorizes through $T$. Since $K
\subseteq A^*$, we can translate the rules $\ell \to r$ in $T$ over
$K^*$ to rules $c \ell \to c r$ over $A^*$. This leads to the set of
$T'$-rules over $A^*$. The letter $c$ at the beginning of the
$T'$-rules is require to shield from $R$-rules. Finally, we show that
$S = R \cup T'$ is the desired system over~$A^*$.

\begin{theorem}\label{thm:group}
  Let $(A,\Abs{\cdot})$ be a weighted alphabet and let $\varphi : A^*
  \to G$ be a homomorphism to a finite group $G$. Then there exists a
  weighted Church-Rosser system $S$ with finite index such that
  $\varphi$ factorizes through $S$.
\end{theorem}

\begin{proof}
  In the following $n$ denotes the exponent of $G$; this is the least
  positive integer $n$ such that $g^n = 1$ for all $g \in G$.  The
  proof is by induction on the size of the alphabet $A$. If $A =
  \smallset{c}$, then we set $S = \oneset{c^n \to 1}$. Let now $A =
  \oneset{a_0, \ldots, a_s, c}$ and let $a_0$ have minimal weight.  We
  set $B = A \setminus \smallset{c}$. Let
  \begin{equation*}
    \gamma_i = a_{i \bmod s}^{n + \floor{i/s}} \, c.
  \end{equation*}
  Since $A$ and $\oneset{a_0c,\ldots,a_sc, c}$ generate the same
  subgroups of $G$ and since every element $a_j c\in G$ occurs infinitely
  often as some $\gamma_i$, there exists $m > 0$ such that for every
  $g \in G$ there exists a word
  \begin{equation*}
    v_g = \gamma_0^{n_0} \cdots \gamma_m^{n_m} \gamma_0
  \end{equation*}
  with $n_i > 0$ satisfying $\varphi(v_g) = g$ and $\Abs{v_g} -
  \Abs{v_h} < n \Abs{a_0}$ for all $g,h \in G$. The latter property
  relies on $\Abs{\gamma_0} + \Abs{a_0} = \Abs{\gamma_s}$ and pumping
  with $\gamma_0^n$ and $\gamma_s^n$ which both map to the neutral
  element of $G$: Assume $\Abs{v_g} - \Abs{v_h} \geq n \Abs{a_0}$ 
  for some $g,h \in G$.  Then we do the following. All $v_g$ with
  maximal weight are multiplied by $\gamma_0^n$ on the left, and for
  all other words $v_h$ the exponent $n_s$ of $\gamma_s$ is replaced
  by $n_s + n$. After that, the maximal difference $\Abs{v_g} -
  \Abs{v_h}$ has decreased at least by $1$ (and at most by $n
  \Abs{a_0}$). 
%
We can iterate this procedure until the weights of all
  $v_g$ differ less than $n \Abs{a_0}$. Let
  \begin{equation*}
    \Gamma = \oneset{\gamma_0,\ldots,\gamma_m}
  \end{equation*}
  be the generators of the $v_g$.
  By induction there exists a weighted Church-Rosser system $R$
  for the restriction $\varphi : B^* \to G$ satisfying the statement
  of the theorem. By \reflem{lem:irr}, we can assume $\Gamma
  \subseteq \IRR_R(B^*)\,c$. Thus $v_g \in \IRR_R(A^*)$ for all $g \in
  G$. Let 
  \begin{equation*}
    K = \IRR_R(B^*)\,c. 
  \end{equation*}
  The set $K$ is a prefix code in $A^*$.  We consider $K$ as an
  extended alphabet and its elements as extended letters. The weight
  $\Abs{u}$ of $u \in K$ is its weight as a word over $A$. Each
  $\gamma_i$ is a letter in $K$. The homomorphism $\varphi : A^* \to
  G$ can be interpreted as a homomorphism $\varphi : K^* \to G$; it is
  induced by $u \mapsto \varphi(u)$ for $u \in K$.  The length
  lexicographic order on $B^*$ induces a linear order $\leq$ on
  $\IRR_R(B^*)$ and hence also on $K$. Here, we assume $a_0 < \cdots <
  a_s$.  The words $v_g$ can be read as words over the weighted
  alphabet $(K,\Abs{\cdot})$ satisfying the following five properties:
  First, $v_g$ starts with the extended letter
  $\gamma_0$. Second, the last two extended letters of $v_g$ are
  $\gamma_m \gamma_0$.  Third, all extended letters in $v_g$ are in
  non-decreasing order from left to right with respect to $\leq$, with
  the sole exception of the last letter $\gamma_0$ which is smaller
  than its predecessor $\gamma_m$.  The fourth property is that all
  extended letters in $v_g$ have a weight greater than $n \Abs{a_0}$.
  And the last important property is that all differences $\Abs{v_g} -
  \Abs{v_h}$ are smaller than $n \Abs{a_0}$.  Let
  \begin{equation*}
    \Delta = \set{\delta \in K^+}{\delta \in K \text{ or }
      \Abs{\delta} \leq n \Abs{a_0}}.
  \end{equation*}
  Note that $\Delta$ is closed under conjugation, i.e., if $uv \in
  \Delta$ for $u,v \in K^*$, then $vu \in \Delta$.  We can think of
  $\Delta$ as the set of all ``short'' words.  Choose $t \geq n$ such
  that all normal forms $v_g$ have no factor $\delta^{t+n}$ for
  $\delta \in \Delta$ and such that $\Abs{c^{t}} \geq \Abs{u}$ for all
  $u \in K^{2n}$. Note that $c \in \Delta$ has the smallest
    weight among all words in $\Delta$.
  
  The first set of rules
  over the extended alphabet $K$ deals with long repetitions of short
  words: The $\Delta$-rules are
  \begin{equation*}
    T_\Delta = \set{\delta^{t+n} \to
    \delta^{t}}{\delta \in \Delta \text{ and $\delta$ is primitive}}.
  \end{equation*}
  Let $F \subseteq K^*$ contain all words which are a factor of some
  $\delta^+$ for $\delta \in \Delta$ and let $J \subseteq K^+$ be
  minimal such that $K^* J K^* = K^* \setminus F$. By
  \reflem{lem:efin}, we have $J \subseteq K^{2n}$. In particular, $J$
  is finite. Since $J$ and $\Delta$ are disjoint, all words in $J$
  have a weight greater than $n \Abs{a_0}$.  Let $\Omega$ contain all
  $\omega \in J$ such that $\omega \in \Gamma K^*$ implies $\omega
  = \gamma \gamma'$ for some $\gamma > \gamma'$, i.e.,
  \begin{equation*}
    \Omega = J \cap \set{\omega \in K^*}{\omega \not\in \Gamma K^* \text{ or }
      \omega =\gamma \gamma' \text{ for some } \gamma > \gamma'}.
  \end{equation*}
  As we will see below, every sufficiently long word without long
  $\Delta$-repetitions contains a factor $\omega \in \Omega$.

  \begin{claim}\label{clm:long:oo}
    There exists a bound $t' \in \N$ such that every word $u \in K^*$
    with $\Abs{u} \geq t'$ contains a factor $\omega \in \Omega$ or a
    factor of the form $\delta^{t+n}$ for $\delta \in \Delta$.
  \end{claim}
  
  \noindent
  \textit{Proof of \refclm{clm:long:oo}.}
  Let $t'' = (t+n+2) \cdot \max\set{\Abs{v} \in \N}{v \in K}$.  First,
  suppose $u \in K^* \setminus K^* \Gamma K^*$ and $\Abs{u} \geq
  t''$. If $u$ is a factor of $\delta^+$, then $\delta^{n+d}$ is a
  factor of $u$ since $\Abs{\delta} \leq \max\set{\Abs{v} \in \N}{v
    \in K}$. Thus we can assume $u \in K^* \setminus F$. By definition
  of $J$, the word $u$ contains a factor $\omega \in J$. We have
  $\omega \in \Omega$ because $u$ (and thus $\omega$) has no factor in
  $\Gamma$.

  If $u \in K^* b \gamma K^*$ for $b \in K \setminus \Gamma$ and
  $\gamma \in \Gamma$, then $u$ contains a factor $\omega = b \gamma
  \in \Omega$. Similarly, if $u \in K^* \gamma \gamma' K^*$ for
  $\gamma,\gamma' \in \Gamma$ and $\gamma > \gamma'$, then $u$
  contains a factor $\omega = \gamma \gamma' \in \Omega$. Thus, if $u
  \in K^* \Gamma K^*$, then we can assume $u = \gamma_{i_1} \cdots
  \gamma_{i_k} u'$ with 
  \begin{itemize}
  \item $\gamma_{i_j} \in \Gamma$ and $\gamma_{i_1}
    \leq \cdots \leq \gamma_{i_k}$, and
  \item $u' \not\in K^* \Gamma K^*$ and $\Abs{u'} < t''$.
  \end{itemize}
  We set $t' = (t+n-1) \cdot \abs{\Gamma} \cdot \max\set{\Abs{v} \in
    \N}{v \in \Gamma} + 1 + t''$.  If $\Abs{u} \geq t'$, then $k \geq
  (t+n-1) \cdot \abs{\Gamma} + 1$. By the pigeon hole principle, there
  exists $\gamma \in \oneset{\gamma_{i_1}, \ldots, \gamma_{i_k}}
  \subseteq \Delta$ such that $\gamma^{t+n}$ is a factor of $u$.
  This completes the proof of \refclm{clm:long:oo}.~~~$\diamond$

  \medskip

  Since $\Delta$ is closed under factors, $u$ contains no factor of
  the form $\delta^{t+n}$ for $\delta \in \Delta$ if and only if $u
  \in \IRR_{T_\Delta}(K^*)$. In particular, it is no restriction to
  only allow primitive words from $\Delta$ in the rules $T_\Delta$.
  Every sufficiently long word $u'$ can be written as $u' =
  u_1 \cdots u_k$ with $\Abs{u_i} \geq t'$ and $k$ sufficiently large.
  Thus, by repeatedly applying \refclm{clm:long:oo}, there exists a
  non-negative integer $d_\Omega$ such that every word $u' \in
  \IRR_{T_\Delta}(K^*)$ with $\Abs{u'} \geq t_\Omega$ contains two
  occurrences of the same $\omega \in \Omega$ which are far
  apart. More precisely, $u'$ has a factor $\omega\, u\, \omega$ with
  $\Abs{u} > \Abs{v_g}$ for all $g \in G$.

  This suggests rules of the form $\omega \, u \, \omega \to \omega\,
  v_{\varphi(u)} \,\omega$; but in order to ensure confluence we have
  to limit their use.  For this purpose, we equip $\Omega$ with a
  linear order $\preceq$ such that $\gamma_m \gamma_0$ is the smallest
  element, and every element in $\Omega \cap K^+ \gamma_0$ is smaller
  than all elements in $\Omega \setminus K^+ \gamma_0$. By making
  $t_\Omega$ bigger, we can assume that every word $u'$ with $\Abs{u'}
  \geq t_\Omega$ contains a factor $\omega\, u\, \omega$ such that
  \begin{itemize}
  \item $\Abs{u} > \Abs{v_g}$ for all $g \in G$, and
  \item for every factor $\omega' \in \Omega$ of
    $\omega\, u\, \omega$ we have $\omega' \preceq \omega$.
  \end{itemize}
  The following claim is one of the main reasons for using the above
  definition of the normal forms $v_g$, and also for excluding all
  words $\omega \in \Gamma K^*$ in the definition of~$\Omega$ except
  for $\omega = \gamma\gamma' \in \Gamma^2$ with $\gamma > \gamma'$.

  \begin{claim}\label{clm:oovoo}
    Let $\omega,\omega' \in \Omega$ and $g \in G$. If $\omega \, v_g
    \, \omega \in K^* \omega' K^*$, then $\omega' \preceq \omega$.
  \end{claim}
  
  \noindent
  \textit{Proof of \refclm{clm:oovoo}.}
  All normal forms $v_g$ have $\gamma_m \gamma_0$ as a suffix. In
  addition, the word $\gamma_m \gamma_0$ is the only element in
  $\Omega$ which is a factor of some $v_g$ for $g \in G$. The reason
  is that all other letters in $v_g$ are in non-decreasing order
  whereas all $\gamma \gamma' \in \Omega$ are in decreasing order.  In
  particular, if $\gamma_m \gamma_0 \, v_g \, \gamma_m \gamma_0 \in
  K^* \omega' K^*$ for $\omega' \in \Omega$, then $\omega' = \gamma_m
  \gamma_0$, i.e., $\gamma_m \gamma_0$ is the only factor of $\gamma_m
  \gamma_0 \, v_g \, \gamma_m \gamma_0$ which is in $\Omega$.
  
  Let now $\omega = b \gamma_0$ for $b \in K \setminus
  \smallset{\gamma_0}$. Note that $\omega \in \Omega$ and that all
  elements in $\Omega \cap K^+ \gamma_0$ have this form.  Then the set of
  factors of $\omega v_g \omega$ which are in $\Omega$ is
  $\oneset{\gamma_m \gamma_0, \omega}$. Since $\gamma_m \gamma_0$ is
  the smallest element with respect to $\preceq$, each of them
  satisfies the claim.

  Next, suppose $\omega \in K^+ b$ for $b \in K \setminus
  \smallset{\gamma_0}$. Then the set of factors of $\omega v_g \omega$
  which are in $\Omega$ is $\oneset{\gamma_m \gamma_0, b \gamma_0,
    \omega}$. Since every element ending with $\gamma_0$ is smaller than
  any other element in $\Omega$, the claim also holds in this case.
  This completes the proof of \refclm{clm:oovoo}.~~~$\diamond$

  \medskip

  We are now ready to define the second set of rules over the extended
  alphabet $K$. They are reducing long words without long repetitions
  of words in $\Delta$. We set
  \begin{equation*}
    T_\Omega' = \set{\omega \, u \, \omega \to \omega\,
      v_{\varphi(u)} \,\omega}{
      \parbox{7.1cm}{$\Abs{v_{\varphi(u)}} 
        < \Abs{u} \leq
      t_\Omega \text{ and } $ \\ 
      $\omega \,u \,\omega \text{ has no
        factor } \omega' \in \Omega \text{ with } \omega \prec
      \omega'$}}.
  \end{equation*}
  Whenever there is a shorter rule in $ T_{\Omega}' \cup T_\Delta$
  then we want to give preference to this shorter rule. Thus the
  $\Omega$-rules are
  \begin{equation*}
    T_\Omega = \set{\ell \to r \in T_\Omega'}{
      \parbox{6.2cm}{$\text{there is no rule } 
      \ell' \to r' \in T_{\Omega}' \cup T_\Delta$ \\
      $\text{such that } 
      \ell' \text{ is a proper factor of } \ell$}}.
  \end{equation*}
  Let now
  \begin{equation*}
    T = T_\Delta \cup T_\Omega \, .
  \end{equation*}

  \begin{claim}\label{clm:confl}
    The system $T$ is locally confluent over $K^*$.
  \end{claim}

  \noindent
  \textit{Proof of \refclm{clm:confl}.} The system $T_\Delta$ is
  confluent by \reflem{lem:econfl}. Suppose we can apply two rules
  $\ell \to r \in T_\Omega$ and $\ell' \to r' \in T_\Delta$. Then
  $\ell'$ is not a factor of $\ell$. Let $\ell = \omega u 
  \omega$. Since $\omega$ is not a factor of $\ell'$, it is possible
  to first apply $\ell \to r$ and then apply $\ell' \to r'$.
  Moreover, by choice of $d$ we have $\Abs{\omega} \leq
  \Abs{r'}$. Thus we also can first apply $\ell' \to r'$ and then
  $\ell \to r$. 

  If $u \in \IRR_{T_\Delta}(K^*)$ and $u \RA{T_\Omega} v$, then $v \in
  \IRR_{T_\Delta}(K^*)$ by definition of the normal forms $v_g$ and
  the set $\Omega$.
  Thus, it remains to show that $T_\Omega$ is locally confluent on
  $\IRR_{T_\Delta}(K^*)$. By minimality of $J$, no $\omega \in \Omega$
  is a proper factor of another word $\omega' \in \Omega$. Let $\omega
  u \omega \to r$ and $\omega' u' \omega' \to r'$ be two
  $\Omega$-rules with $\omega \neq \omega'$. By construction of
  $T'_\Omega$, the left sides of both rules can overlap at most $\min
  \oneset{\abs{\omega}, \abs{\omega'}} - 1$ positions. Thus the two
  rules can always be applied independently of one another.

  Let now $\omega u \omega \to \omega v_g \omega$ and $\omega u'
  \omega \to \omega v_h \omega$ be two $\Omega$-rules. By construction
  of $T_\Omega$, neither is $\omega u' \omega$ a proper factor of
  $\omega u \omega$ nor vice versa. If $x \omega = \omega y$ for some
  $x,y \in K^+$ with $\Abs{x} \leq n \Abs{a_0}$, then $x \in \Delta$
  and $\omega$ is a prefix of $x^+$ which contradicts the definition
  of $J \subseteq K^* \setminus F$. Therefore, whenever $x \omega =
  \omega y$ for $x, y \in K^+$ then $\Abs{x} > n \Abs{a_0}$ and
  $\Abs{y} > n \Abs{a_0}$.
  Suppose now $x \omega u \omega = \omega u' \omega y = \omega u''
  \omega$ for $x,y \in K^+$. If $\abs{x} \geq \abs{\omega u}$, then
  the two rules can be applied independently of one another. Thus let
  $\abs{x} < \abs{\omega u}$.  As seen before, we have $\Abs{x} > n
  \Abs{a_0}$ and $\Abs{y} > n \Abs{a_0}$. We will show
  \begin{equation*}
    x \,\omega\, v_g\, \omega \;\RAS*{T_\Omega}\; 
    \omega\, v_{\varphi(u'')}\, \omega 
    \;\LAS*{T_\Omega}\; \omega \, v_h \, \omega \, y.
  \end{equation*}
  If $x \,\omega\, v_g\, \omega \in K^* \omega' K^*$ or $\omega \, v_h
  \, \omega \, y \in K^* \omega' K^*$, then by \refclm{clm:oovoo} we
  have $\omega' \preceq \omega$.
  We can write $x \omega = \omega x'$. Since $\Abs{x'} = \Abs{x} > n
  \Abs{a_0}$, we have $\Abs{x' v_g} > n \Abs{a_0} + \Abs{v_g} >
  \Abs{v_{g'}}$ for every $g' \in G$. This relies on the fact that the
  weights all normal forms $v_{g'}$ differ less than $n \Abs{a_0}$.
  This shows that the weight of $x' v_g$ is sufficiently high. If
  $\Abs{x' v_g} > t_\Omega$, then by \refclm{clm:long:oo} we have $x'
  v_g \RAS*{T_\Omega} x''$ such that $\Abs{v_{g'}} < \Abs{x''} \leq
  t_\Omega$ for every $g' \in G$. Therefore, without loss of
  generality we can assume that the weight of $x' v_g$ is not too
  high, i.e., $\Abs{x' v_g} \leq t_\Omega$. Since $\varphi(x' v_g) =
  \varphi(u'')$, we have $x \omega v_g \omega \RAS*{T_\Omega} \omega
  v_{\varphi(u'')} \omega$. Similarly, $\omega v_h \omega y
  \RAS*{T_\Omega} \omega v_{\varphi(u'')} \omega$.
  This completes the proof of \refclm{clm:confl}.~~~$\diamond$

  \medskip

  Since all rules in $T$ are weight-reducing, local confluence implies
  confluence. Moreover, all rules $\ell \to r$
  in $T$ satisfy $\varphi(\ell) = \varphi(r)$. We conclude that $T$ is
  a weighted Church-Rosser system such that $K^* / T$ is finite
  and $\varphi : K^* \to G$ factorizes through $T$. Remember that
  every element in $K^*$ can be read as a sequence of elements in
  $A^*$.  Thus every $u \in K^*$ can be interpreted as a word $u \in
  A^*$.  We use this interpretation in order to apply the rules in $T$
  to words in $A^*$; but in order to not destroy $K$-letters when
  applying rules in $R$, we have to guard the first $K$-letter
  of every $T$-rule by appending the letter $c$. This leads to the
  system
  \begin{equation*}
    T' = \set{c\ell \to cr \in A^* \times A^*}{\ell \to r \in T}.
  \end{equation*}
  Combining the rules $R$ over the alphabet $B$ with the
  $T'$-rules yields
  \begin{equation*}
    S = R \cup T'.
  \end{equation*}
  Since left sides of $R$-rules and of $T'$-rules can not overlap, the
  system $S$ is confluent. By definition, each $S$-rule is
  weight-reducing. This means that $S$ is a weighted
  Church-Rosser system. We have
  \begin{align*}
    \IRR_S(A^*) \ = \ 
    \IRR_R(B^*) \;\cup \;
    \IRR_R(B^*) \cdot \IRR_{T'}\Big( 
      c\big(\IRR_R(B^*) c\big)^* \Big) \cdot \IRR_R(B^*).
  \end{align*}
  Therefore $\IRR_S(A^*)$ and $A^* / S$ are finite.  Since all rules
  $\ell \to r$ in $S$ satisfy $\varphi(\ell) = \varphi(r)$, the
  homomorphism $\varphi$ factorizes through $S$.
%
\end{proof}
%

\section{Arbitrary Finite Monoids}\label{sec:arbitrary}

This section contains the main result of this paper. We show that
every homomorphism $\varphi : A^* \to M$ to finite monoid factorizes
through a weighted Church-Rosser system $S$ with finite index.
The proof relies on \refthm{thm:group} and on a construction called
local divisors.

\subsection{Local divisors}

The notion of {\em local divisor} has turned out to be a rather
powerful tool when using inductive proofs for finite monoids, see
e.g.~\cite{dg08SIWT:short,DiekertKS11,DiekertKW12tcs}. The same is
true in this paper. The definition of a local divisor is as follows:
Let $M$ be a monoid and let $c \in M$.  We equip $cM \cap Mc$ with a
monoid structure by introducing a new multiplication~$\circ$ as
follows:
\begin{equation*}
  xc \circ cy = xcy.
\end{equation*}
It is straightforward to see that $\circ$ is well-defined and $(cM
\cap Mc, \circ)$ is a monoid with neutral element $c$. 
  
The following observation is crucial.  If $1 \in {cM \cap Mc}$, then
$c$ is a unit. Thus if the monoid $M$ is finite and $c$ is not a unit,
then $\abs{cM \cap Mc} < \abs{M}$.  The set $M' = \set{x}{cx \in Mc}$
is a submonoid of $M$, and $c{\cdot}: M' \to cM \cap Mc : x \mapsto
cx$ is a surjective homomorphism. Since $(cM \cap Mc,
\circ)$ is the homomorphic image of a submonoid, it is a divisor of
$M$.  We therefore call $(cM \cap Mc, \circ)$ the {\em local divisor}
of $M$ at $c$.

\subsection{The main result}

We are now ready to prove our main result: Every homomorphism $\varphi
: A^* \to M$ to a finite monoid factorizes through a weighted
Church-Rosser system $S$ with finite index. The proof uses induction
on the size of $M$ and the size of $A$. If $\varphi(A^*)$ is a group,
then we apply \refthm{thm:group}; and if $\varphi(A^*)$ is not a
group, then we find a letter $c \in A$ such that $c$ is not a
unit. Thus in this case we can use local divisors.

\begin{theorem}\label{thm:main}
  Let $(A,\Abs{\cdot})$ be a weighted alphabet and let $\varphi : A^*
  \to M$ be a homomorphism to a finite monoid $M$. Then there exists a
  weighted Church-Rosser system $S$ of finite index such that
  $\varphi$ factorizes through $S$.
\end{theorem}

\begin{proof}
  The proof is by induction on $(\abs{M},\abs{A})$ with lexicographic
  order. If $\varphi(A^*)$ is a group, then the claim follows by
  \refthm{thm:group}. If $\varphi(A^*)$ is not a group, then there
  exists $c \in A$ such that $\varphi(c)$ is not a unit.  Let $B = A
  \setminus \smallset{c}$. By induction on the size of the alphabet
  there exists a weighted Church-Rosser system $R$ for the
  restriction $\varphi : B^* \to M$ satisfying the statement of the
  theorem. Let
  \begin{equation*}
    K = \IRR_R(B^*) c.
  \end{equation*}
  We consider the prefix code $K$ as a weighted alphabet. The weight
  of a letter $uc \in K$ is the weight $\Abs{uc}$ when read as a word
  over the weighted alphabet $(A,\Abs{\cdot})$. Let $M_c = \varphi(c)
  M \cap M \varphi(c)$ be the local divisor of $M$ at $\varphi(c)$.
  We let $\psi : K^* \to M_c$ be the homomorphism induced by $\psi(uc)
  = \varphi(cuc)$ for $uc \in K$. By induction on the size of the
  monoid there exists a weighted Church-Rosser system $T
  \subseteq K^* \times K^*$ for $\psi$ satisfying the statement of the
  theorem. Suppose $\psi(\ell) = \psi(r)$ for $\ell,r \in K^*$ and let
  $\ell = u_1 c \cdots u_j c$ and $r = v_1 c \cdots v_k c$ with
  $u_i,v_i \in \IRR_R(B^*)$. Then
  \begin{align*}
    \varphi(c \ell) 
    &= \varphi(cu_1c) \circ \cdots \circ \varphi(cu_jc) \\
    &= \psi(u_1 c) \circ \cdots \circ \psi(u_j c) \\
    &= \psi(\ell) = \psi(r) = \varphi(c r).
  \end{align*}
  This means that every $T$-rule $\ell \to r$ yields an
  $\varphi$-invariant rule $c\ell \to cr$.  Thus we can transform the
  system $T \subseteq K^* \times K^*$ for $\psi$ into a system $T'
  \subseteq A^* \times A^*$ for $\varphi$ by
  \begin{equation*}
    T' = \set{c\ell \to cr \in A^* \times A^*}{\ell \to r \in T}.
  \end{equation*}
  Since $T$ is confluent and weight-reducing over $K^*$, the system
  $T'$ is confluent and weight-reducing over $A^*$. Combining $R$ and
  $T'$ leads to
  \begin{equation*}
    S = R \cup T'.
  \end{equation*}
  The left sides of a rule in $R$ and a rule in $T'$ cannot overlap.
  Therefore, $S$ is a weighted Church-Rosser system such that
  $\varphi$ factorizes through $A^* / S$. Suppose that every word in
  $\IRR_T(K^*)$ has length at most $k$. Here, the length is over the
  extended alphabet $K$. Similarly, let every word in $\IRR_R(B^*)$
  have length at most~$m$. Then
  \begin{equation*}
    \IRR_S(A^*) \subseteq \set{u_0 c u_1 \cdots c u_{k'+1}}{ 
      u_i \in \IRR_R(B^*), \; k' \leq k}
  \end{equation*}
  and every word in $\IRR_S(A^*)$ has length at most $(k+2)m$. In
  particular $\IRR_S(A^*)$ and $A^* / S$ are finite.
%
\end{proof}

The following corollary is a straightforward translation of the result
in \refthm{thm:main} about homomorphisms to a statement about regular
languages.

\begin{corollary}\label{cor:main}
  A language $L \subseteq A^*$ is regular if and only if there exists
  a Church-Rosser system $S$ of finite index such that $L =
  \bigcup_{u \in L} [u]_S$.
\end{corollary}

\begin{proof}
  If $L$ is regular, then there exists a homomorphism $\varphi : A^*
  \to M$ recognizing $L$. By \refthm{thm:main} there exists a finite
  Church-Rosser system $S$ of finite index such that $\varphi$
  factorizes through $S$. The latter property implies $\varphi^{-1}(x)
  = \bigcup_{u \in \varphi^{-1}(x)} [u]_S$ for every $x \in M$. Thus
  $L = \bigcup_{x \in \varphi(L)} \varphi^{-1}(x) = \bigcup_{u \in L}
  [u]_S$.
  The converse is trivial. 
  %
\end{proof}

In particular, we see that all regular languages are Church-Rosser
congruential.


\begin{thebibliography}{10}

\bibitem{bo93springer}
R.~Book and F.~Otto.
\newblock {\em String-Rewriting Systems}.
\newblock Springer-Verlag, 1993.

\bibitem{BuntrockO98}
G.~Buntrock and F.~Otto.
\newblock Growing context-sensitive languages and {C}hurch-{R}osser languages.
\newblock {\em Information and Computation}, 141:1--36, 1998.

\bibitem{dg08SIWT:short}
V.~Diekert and P.~Gastin.
\newblock First-order definable languages.
\newblock In {\em Logic and Automata: History and Perspectives}, Texts in Logic
  and Games, pages 261--306. Amsterdam University Press, 2008.

\bibitem{DiekertKS11}
V.~{Diekert}, M.~{Kufleitner}, and B.~{Steinberg}.
\newblock {The {K}rohn-{R}hodes Theorem and Local Divisors}.
\newblock {\em ArXiv e-prints}, Nov. 2011.

\bibitem{DiekertKW12tcs}
V.~Diekert, M.~Kufleitner, and P.~Weil.
\newblock Star-free languages are {C}hurch-{R}osser congruential.
\newblock {\em Theoretical Computer Science}, 2012.
\newblock DOI: 10.1016/j.tcs.2012.01.028.

\bibitem{jan88eatcs}
M.~Jantzen.
\newblock {\em Confluent String Rewriting}, volume~14 of {\em EATCS Monographs
  on Theoretical Computer Science}.
\newblock Springer-Verlag, 1988.

\bibitem{lot83}
M.~Lothaire.
\newblock {\em Combinatorics on Words}, volume~17 of {\em Encyclopedia of
  Mathematics and its Applications}.
\newblock Addison-Wesley, Reading, MA, 1983.
\newblock Reprinted by {\em Cambridge University Press}, 1997.

\bibitem{McNaughtonNO88}
R.~McNaughton, P.~Narendran, and F.~Otto.
\newblock {C}hurch-{R}osser {T}hue systems and formal languages.
\newblock {\em J. ACM}, 35(2):324--344, 1988.

\bibitem{Narendran84phd}
P.~Narendran.
\newblock {\em {C}hurch-{R}osser and related {T}hue systems}.
\newblock Doctoral dissertation, Dept.~of Mathematical Sciences, Rensselaer
  Polytechnic Institute, Troy, NY, USA, 1984.

\bibitem{NiemannPhD02}
G.~Niemann.
\newblock {\em Church-{R}osser Languages and Related Classes}.
\newblock Kassel University Press, 2002.
\newblock PhD thesis.

\bibitem{NiemannO05}
G.~Niemann and F.~Otto.
\newblock The {C}hurch-{R}osser languages are the deterministic variants of the
  growing conext-sensitive languages.
\newblock {\em Inf. Comput.}, 197:1--21, 2005.

\bibitem{NiemannW02}
G.~Niemann and J.~Waldmann.
\newblock Some regular languages that are {C}hurch-{R}osser congruential.
\newblock In {\em Revised Papers from the 5th International Conference on
  Developments in Language Theory}, DLT '01, pages 330--339, London, UK, 2002.
  Springer-Verlag.

\bibitem{reinhardtT03}
K.~Reinhardt and D.~Th{\'e}rien.
\newblock Some more regular languages that are {C}hurch {R}osser congruential.
\newblock In {\em 13. Theorietag, Automaten und Formale Sprachen, Herrsching,
  Germany}, pages 97--103, 2003.

\bibitem{Woinowski01Diss}
J.~R. Woinowski.
\newblock {\em {C}hurch-{R}osser Languages and Their Application to Parsing
  Problems}.
\newblock PhD thesis, TU Darmstadt, 2001.

\bibitem{Woinowski03IC}
J.~R. Woinowski.
\newblock The context-splittable normal form for {C}hurch-{R}osser language
  systems.
\newblock {\em Inf. Comput.}, 183:245--274, 2003.

\end{thebibliography}

{\small
\newcommand{\Ju}{Ju}\newcommand{\Ph}{Ph}\newcommand{\Th}{Th}\newcommand{\Ch}{Ch}\newcommand{\Yu}{Yu}\newcommand{\Zh}{Zh}

}

\end{document}